\documentclass{article}

\usepackage{etex} 
\reserveinserts{28}

\usepackage{bbm}
\usepackage{mathrsfs}
\usepackage{latexsym}

\usepackage{suffix}

\usepackage{xfrac} 
\usepackage[style=english,english=british]{csquotes}

\usepackage{amsmath, amsthm, amssymb, amsfonts} 
\usepackage{thmtools, thm-restate}
\usepackage{mathtools}

\usepackage[full,small]{complexity}

\usepackage[final,notref,color]{showkeys} 

\usepackage{xargs}
\usepackage{afterpage} 
                       \usepackage{authblk}

\PassOptionsToPackage{dvipsnames,usenames,table}{xcolor}
\usepackage{tikz}
\usetikzlibrary{calc}

\usepackage{booktabs} 
\usepackage{longtable}
\usepackage{float}
\usepackage{wrapfig}
\usepackage{floatflt}
\usepackage{framed}
\usepackage{ctable}
\usepackage{dcolumn}
\usepackage{adjustbox}

\usepackage{enumitem}
\usepackage{xspace}
\usepackage{xifthen}
\usepackage{fixltx2e}
\usepackage{url}  
\usepackage{scrtime}

\usepackage{xkvltxp}

\makeatletter
\newcommand{\raisemath}[1]{\mathpalette{\raisem@th{#1}}}
\newcommand{\raisem@th}[3]{\raisebox{#1}{$#2#3$}}
\makeatother

\newcommand\restr[2]{{
  \left.\kern-\nulldelimiterspace 
  #1 
  \vphantom{\big|} 
  \right|_{#2} 
  }}

\def\adm_#1{ \operatorname{adm}_{#1} }

\renewcommand{\le}{\leqslant}
\renewcommand{\leq}{\leqslant}
\renewcommand{\ge}{\geqslant}
\renewcommand{\geq}{\geqslant}

\renewcommand{\epsilon}{\varepsilon}

\usepackage{stmaryrd}

\usepackage{pifont}
\newcommand{\yaay}{\kern4pt \ding{51} \kern-8pt \ding{51}}%
\theoremstyle{plain}
\newtheorem{lemma}{Lemma}

\newtheorem{theorem}{Theorem}
\newtheorem{corollary}{Corollary}

\newtheorem*{claim}{Claim}

\newtheoremstyle{case}
  {\topsep}   
  {\topsep}   
  {}  
  {\parindent}       
  {\bfseries} 
  {\normalfont.}         
  {5pt plus 1pt minus 1pt} 
  {#1 #2: {\normalfont #3}}          

\theoremstyle{case}

\numberwithin{subcase}{case}
\makeatletter
\@addtoreset{case}{lemma}
\@addtoreset{case}{theorem}
\@addtoreset{case}{corollary}
\makeatother

\theoremstyle{definition}

\newtheorem{definition}{Definition}

\newcommand*\varrule[1][0.4pt]{\leavevmode\leaders\hrule height#1\hfill\kern0pt}

\setlist[1]{labelindent=\parindent,leftmargin=*} 
\setlist{itemsep=0pt}
\setitemize[1]{label={\small\textbullet}}

\newenvironment{tightcenter}
 {\parskip=0pt\par\nopagebreak\centering}
 {\par\noindent\ignorespacesafterend}

\usepackage{ctable}
\newlength{\RoundedBoxWidth}
\newsavebox{\GrayRoundedBox}
\newenvironment{GrayBox}[1]%
   {\setlength{\RoundedBoxWidth}{\textwidth-4.5ex}
    \def\boxheading{#1}
    \begin{lrbox}{\GrayRoundedBox}
       \begin{minipage}{\RoundedBoxWidth}%
   }{%
       \end{minipage}
    \end{lrbox}%
    \begin{tightcenter}%
    \begin{tikzpicture}%
       \node(Text)[draw=black!20,fill=white,rounded corners,%
             inner sep=2ex,text width=\RoundedBoxWidth]%
             {\usebox{\GrayRoundedBox}};
        \coordinate(x) at (current bounding box.north west);
        \node [draw=white,rectangle,inner sep=3pt,anchor=north west,fill=white] 
        at ($(x)+(6pt,.75em)$) {\boxheading};
    \end{tikzpicture}
    \end{tightcenter}\vspace{0pt}%
    \ignorespacesafterend
}    

\newenvironment{problem}[2][]{\noindent\ignorespaces%
                                \FrameSep=6pt%
                                \parindent=0pt%
                \vspace*{-.5em}
                \ifthenelse{\isempty{#1}}{%
                  \begin{GrayBox}{\textsc{#2}}%
                }{%
                  \begin{GrayBox}{\textsc{#2} parametrised by~{#1}}%
                }
                \newcommand\Prob{Question:}%
                \newcommand\Input{Input:}%
                \begin{tabular*}{\textwidth}{@{\hspace{.1em}} >{\itshape} p{1.6cm} p{0.8\textwidth} @{}}%
            }{
                \end{tabular*}%
                \end{GrayBox}%
                \vspace*{-.5em}
                \ignorespacesafterend
            } 

\title{$k$-Distinct In- and Out-Branchings in Digraphs\footnote{A short version of this paper was published in the proceedings of ICALP 2017. Research of Gutin was partially supported by Royal Society Wolfson Research Merit Award.}}
\author[1]{Gregory Gutin}\author[2]{Felix Reidl}\author[1]{Magnus Wahlstr{\"o}m}
\affil[1]{Royal Holloway, University of London, UK}
\affil[2]{North Carolina State University, USA}

\graphicspath{{.,a}{./figures/}}

\begin{document}

\maketitle

\begin{abstract}
  An out-branching and an in-branching of a digraph $D$ are called
  $k$-distinct if each of them has $k$ arcs absent in the other. Bang-Jensen,
  Saurabh and Simonsen (2016)  proved that the problem of deciding whether a
  strongly connected digraph $D$ has $k$-distinct out-branching and in-
  branching is fixed-parameter tractable (FPT) when parameterized by $k$. They
  asked whether the problem remains FPT when extended to arbitrary digraphs.
  Bang-Jensen and Yeo (2008) asked whether the same problem is FPT when the
  out-branching and in-branching have the same root.

  By linking the two problems with the problem of whether a digraph has an
  out-branching with at least $k$ leaves (a leaf is a vertex of out-degree
  zero), we first solve the problem of Bang-Jensen and Yeo (2008). We then
  develop a new digraph decomposition called the rooted cut decomposition and
  using it we prove that the problem of Bang-Jensen et al. (2016) is FPT for
  all digraphs. We believe that the \emph{rooted cut decomposition} will be
  useful for obtaining other results on digraphs.
\end{abstract}

\section{Introduction}
While both undirected and directed graphs are important in many applications,
there are significantly more algorithmic and structural results for undirected
graphs than for directed ones. The main reason is likely to be the fact that
most problems on digraphs are harder than those on undirected graphs. The
situation has begun to change: recently there appeared a number of important
structural results on digraphs, see e.g.
\cite{fradkinJCTB103,kawarabayashiSTOC15,kimJCTB112}. However, the progress
was less pronounced with algorithmic results on digraphs, in particular, in
the area of parameterized algorithms.

In this paper, we introduce a new decomposition for digraphs and show its usefulness by
solving an open problem by Bang-Jensen, 
Saurabh and Simonsen~\cite{bangA76}. 
We believe that our decomposition will prove to be helpful for obtaining further
algorithmic and structural results on digraphs.

A digraph $T$ is an {\em out-tree} (an {\em in-tree}) if $T$ is an oriented
tree with just one vertex $s$ of in-degree zero (out-degree zero). The vertex
$s$ is the {\em root} of $T.$ A vertex $v$ of an out-tree (in-tree) is called a {\em
leaf} if it has out-degree (in-degree) zero. If an out-tree (in-tree) $T$ is a spanning
subgraph of a digraph $D,$ then $T$ is an {\em out-branching} (an {\em 
in-branching}) of $D$. It is well-known that a digraph $D$ contains an out-branching (in-branching)
if and only if $D$ has only one strongly connected component with no incoming
(no outgoing) arc \cite{bang2002}. 

A well-known result in digraph algorithms, due to Edmonds, states that given a
digraph $D$ and a positive integer $\ell$, we can decide whether $D$ has
$\ell$ arc-disjoint out-branchings in polynomial time~\cite{edmonds1973}. The
same result holds for $\ell$ arc-disjoint in-branchings. Inspired by this
fact, it is natural to ask for a ``mixture" of out- and in-branchings: given a
digraph $D$ and a pair $u,v$ of (not necessarily distinct) vertices, decide
whether $D$ has an arc-disjoint out-branching $T^+_u$ rooted at $u$ and 
in-branching $T^-_v$ rooted at $v$. We will call this problem {\sc Arc-Disjoint Branchings}.

Thomassen proved (see  \cite{bangJCT51}) that the problem is NP-complete and
remains NP-complete even if we add the condition that $u=v$. The same result still
holds for digraphs in which the out-degree and in-degree of every vertex
equals two \cite{bangDAM161}. The problem is polynomial-time solvable for
tournaments \cite{bangJCT51} and for acyclic digraphs
\cite{bangJGT42,bercziIPL109}. The single-root special case (i.e., when $u=v$)
of the problem is polynomial time solvable for quasi-transitive
digraphs\footnote{A digraph $D=(V,A)$ is quasi-transitive if for every
$xy,yz\in A$ there is at least one arc between $x$ and $z$, i.e. either $xz\in
A$ or $zx\in A$ or both.} \cite{bangJGT20} and for locally semicomplete
digraphs\footnote{A digraph $D=(V,A)$ is locally semicomplete if for every
$xy,xz\in A$ there is at least one arc between $y$ and $z$ and for every
$yx,zx\in A$ there is at least one arc between $y$ and $z$. Tournaments and
directed cycles are locally semicomplete digraphs.} \cite{bangJGT77}.

An out-branching $T^+$ and an in-branching $T^-$ are called {\em $k$-distinct} if $|A(T^+) \setminus A(T^-) | \geq k$.
Bang-Jensen, Saurabh and Simonsen \cite{bangA76} considered the following parameterization of {\sc Arc-Disjoint Branchings}.

\begin{problem}[$k$]{$k$-Distinct Branchings}
  \Input & A digraph $D$, an integer~$k$. \\
  \Prob & Are there $k$-distinct out-branching $T^+$ and in-branching $T^-$?
\end{problem}

\noindent They proved that {\sc $k$-Distinct Branchings} is fixed-parameter
tractable (FPT)\footnote{Fixed-parameter tractability of {\sc $k$-Distinct
Branchings} means that the problem can be solved by an algorithm of runtime
$O^*(f(k))$, where $O^*$ omits not only constant factors, but also polynomial
ones, and $f$ is an arbitrary computable function. The books
\cite{cygan2015,downey2013} are excellent recent introductions to
parameterized algorithms and complexity.} when $D$ is strongly connected and
conjectured that the same holds when $D$ is an arbitrary digraph. Earlier,
Bang-Jensen and Yeo \cite{bangDAM156} considered the version of {\sc
$k$-Distinct Branchings} where $T^+$ and $T^-$ must have the same
root
and asked whether this version of {\sc $k$-Distinct Branchings}, which we call
{\sc Single-Root $k$-Distinct Branchings}, is FPT.

The first key idea of this paper is to relate  {\sc $k$-Distinct Branchings}
to the problem of deciding whether a digraph has an out-branching with at
least $k$ leaves via a simple lemma (see Lemma~\ref{lemma:leaves-branch}). The
lemma and the following two results on out-branchings with at least $k$
leaves allow us to solve the problem of Bang-Jensen and Yeo \cite{bangDAM156}
and to provide a shorter proof for the above-mentioned result of Bang-Jensen,
Saurabh and Simonsen \cite{bangA76} (see Theorem~\ref{thm:strong}).

\begin{theorem}[\cite{alonSIAMJDM23}]\label{thm:alon}
  Let $D$ be a strongly connected digraph. If $D$ has no out-branching with at least $k$ leaves, then the (undirected) pathwidth of $D$ is bounded by $O(k\log k)$.
\end{theorem}

\begin{theorem}[\cite{daligaultJCSS76,kneisA61}] \label{thm:kleaf}
  We can decide whether a digraph $D$ has an out-branching with at least $k$ leaves in time\footnote{The 
  algorithm of \cite{kneisA61} runs in time $O^*(4^k)$ and its modification in \cite{daligaultJCSS76} in time $O^*(3.72^k)$.} 
   $O^*(3.72^k)$.
\end{theorem}

\noindent
The general case of {\sc $k$-Distinct Branchings} seems to be much more
complicated. We first introduce a version of {\sc $k$-Distinct Branchings}
called {\sc $k$-Rooted Distinct Branchings}, where the roots $s$ and $t$ of
$T^+$ and $T^-$ are fixed, and add arc $ts$ to $D$ (provided the arc is not in
$D$) to make $D$ strongly connected. This introduces a complication: we may
end up in a situation where $D$ has an out-branching with many leaves, and
thereby potentially unbounded pathwidth, but the root of the out-branching is
not $s$. To deal with this situation, our goal will be to
\emph{reconfigure} the out-branching into an out-branching rooted at $s$. In
order to reason about this process, we develop a new digraph decomposition we
call the \emph{rooted cut decomposition}. The cut decomposition of a digraph
$D$ rooted at a given vertex $r$ consists of a tree $\hat T$ rooted at $r$
whose nodes are some vertices of $D$ and subsets of vertices of $D$ called
{\em diblocks} associated with the nodes of $\hat T$.

Our strategy is now as follows. If $\hat T$ is
\emph{shallow} (i.e., it has bounded height), then any out-branching
with sufficiently many leaves can be turned into an out-branching rooted at $s$
without losing too many of the leaves. On the other hand, if $\hat T$ contains a path from the root of $\hat T$ with sufficiently many non-degenerate
diblocks (diblocks with at least three vertices), then we are able to show
immediately that the instance is positive.
The remaining and most difficult issue is to deal with digraphs with
decomposition trees that contain long paths of diblocks with only two
vertices, called \emph{degenerate} diblocks. In this case, we employ 
two reduction rules which lead to decomposition trees of bounded height.

The paper is organized as follows. In the next section, we provide some
terminology and notation on digraphs used in this paper. In
Section~\ref{sec:scd}, we prove Theorem~\ref{thm:strong}.
Section~\ref{sec:any} is devoted to proving that {\sc Rooted $k$-Distinct
Branchings}  is FPT for all digraphs using cut decomposition and
Theorems~\ref{thm:alon} and~\ref{thm:kleaf}. We conclude the paper in
Section~\ref{sec:conclusion}, where some open parameterized problems on
digraphs are mentioned. 

\section{Terminology and Notation}\label{sec:term}

Let us recall some basic terminology of
digraph theory, see \cite{bang2002}. A digraph $D$ is {\em strongly connected}
({\em connected}) if there is a directed (oriented) path from $x$ to $y$ for
every ordered pair $x,y$ of vertices of $D$. Equivalently, $D$ is connected
if the underlying graph of $D$ is connected. A vertex $v$ is a {\em source}
({\em sink}) if its in-degree (out-degree) is equal to zero. It is well-known
that every acyclic digraph has a source and a sink \cite{bang2002}.

In this paper, we exclusively work with digraphs, therefore we assume all our
graphs, paths, and trees to be directed unless otherwise noted.  For a path~$P = x_1x_2 \ldots x_k$ of length~$k-1$
we will employ the following notation for subpaths of~$P$:
$P[x_i,x_j] := x_i\ldots x_j$ for~$1 \leq i \leq j \leq k$ is
the \emph{infix} of~$P$ from~$x_i$ to~$x_j$. For paths~$P_1 := x_1\ldots x_kv$
and~$P_2 := vy_1\ldots y_\ell$ we denote by~$P_1P_2 := x_1\ldots x_k v y_1\ldots y_\ell$
their \emph{concatenation}. For rooted
trees~$T$ and some vertex~$x \in T$, $T_x$ stands for the
subtree of~$T$ rooted at~$x$ (see Figure~\ref{fig:fins}). 

\begin{figure}
  \centering\includegraphics[scale=.7]{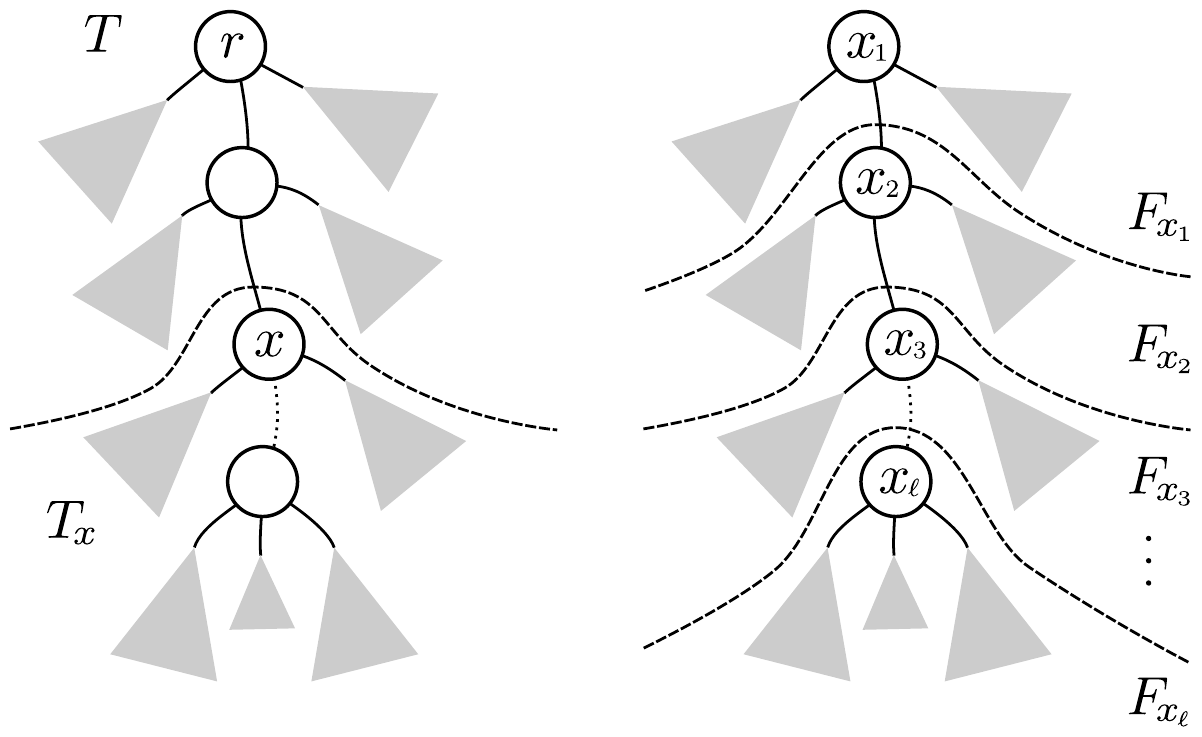}
  \caption{\label{fig:fins}%
    Subtree notation~$T_x$ for~$x \in T$ (left) and the fins~$F_{x_1}, \ldots, F_{x_\ell}$
    for a path~$x_1\ldots x_\ell$ in~$T$ (right).
  }\vspace*{-.5em}
\end{figure}

We will frequently partition the nodes of a tree around a path in the following
sense ({\em cf.} Figure~\ref{fig:fins}):
  Let~$T$ be a tree rooted at~$r$ and~$P = x_1\ldots x_\ell$ a path from~$r =
  x_1$ to some node~$x_\ell \in T$. The \emph{fins} of~$P$ are the
  sets~$\{F_{x_i}\}_{x_i \in P}$ defined as
$
      F_{x_i} := V(T_{x_i}) \setminus V(T_{x_{i+1}})~\text{for~$i < \ell$ and } 
      F_{x_\ell} := V(T_{x_\ell}).
$

\begin{definition}[Bi-reachable Vertex]  A vertex $v$ of a digraph $D$ is \emph{bi-reachable} from a vertex~$r$
if there exist two internally vertex-disjoint paths from~$r$ to~$v$.
\end{definition}

\noindent
Given a digraph~$D$ and a vertex~$r$, we can compute the set of 
vertices that are bi-reachable from~$r$ in polynomial time using 
network flows.

\section{Strongly Connected Digraphs}\label{sec:scd}

Let us prove a simple fact on a link between out/in-branchings with many
leaves and {\sc $k$-Distinct Branchings}, which together with a structural
result of Alon {\em et al.} \cite{alonSIAMJDM23} and an algorithmic result for the {\sc Maximum 
Leaf Out-branching} problem \cite{daligaultJCSS76,kneisA61} gives a short proof that
both versions of {\sc $k$-Distinct Branchings} are FPT for strongly connected digraphs.

\begin{lemma}\label{lemma:leaves-branch}
Let $D$ be a digraph containing an out-branching and an in-branching.  
  If $D$ contains an out-branching (in-branching)~$T$ with at least~$k+1$ leaves, then
  every in-branching (out-branching)~$T'$ of~$D$ is $k$-distinct from~$T$.
\end{lemma}
\begin{proof}
We will consider only the case when $T$ is an out-branching since the other
case can be treated similarly.  Let $T'$ be an in-branching of $D$ and let
$L$ be the set of all leaves of $T$ apart from the one which is the root of
$T'$. Observe that all vertices of $L$ have outgoing arcs in $T'$ and since
in $T$ the incoming arcs of $L$ are the only arcs incident to $L$ in $T$,
the sets of the outgoing arcs in $T'$ and incoming arcs in $T$ do not
intersect.
\end{proof}

\noindent
From the next section, the following problem will be of our main interest.
The problem \textsc{$k$-Distinct Branchings} in which $T^+$ and $T^-$ must be rooted at given vertices $s$ and $t$, respectively, will be called
the \textsc{Rooted $k$-Distinct Branchings} problem. 
We will use the following standard dynamic programming result (see, e.g., \cite{bangA76}).

\begin{lemma}\label{lem:dp}
  Let $H$ be a digraph of (undirected) treewidth  $\tau$. Then {\sc $k$-Distinct
  Branchings}, {\sc Single-Root $k$-Distinct Branchings} as well as \textsc{Rooted $k$-Distinct Branchings} 
  on $H$ can be solved
  in time $O^*(2^{O(\tau\log \tau)})$.
\end{lemma}

\noindent Note that if a digraph $D$ is a positive instance of {\sc Single-Root $k$-Distinct Branchings} 
then $D$ must be strongly connected as an out-branching and an in-branching rooted 
at the same vertex form a strongly connected subgraph of $D$. Thus, the following theorem, in particular, solves
the problem of Bang-Jensen and Yeo mentioned above.

\begin{theorem}\label{thm:strong}
  Both {\sc $k$-Distinct Branchings} and {\sc Single-Root $k$-Distinct Branchings}
  on strongly connected digraphs can be solved in time $O^*(2^{O(k\log^2 k)})$.
\end{theorem}
\begin{proof}
  The proof is essentially the same for both problems and we will give it for
  {\sc Single-Root $k$-Distinct Branchings}. Let $D$ be an input strongly
  connected digraph. By Theorem~\ref{thm:kleaf} using an $O^*(3.72^k)$-time
  algorithm we can find an out-branching $T^+$ with at least $k+1$ leaves, or
  decide that $D$ has no such out-branching. If $T^+$ is found, the instance of
  {\sc Single-Root $k$-Distinct Branchings} is positive by 
  Lemma~\ref{lemma:leaves-branch} as any in-branching $T^-$ of $D$ is $k$-distinct from $T^+$.
  In particular, we may assume that $T^-$ has the same root as $T^+$ (a strongly
  connected digraph has an in-branching rooted at any vertex).
  Now suppose that $T^+$ does not exist. Then, by  Theorem~\ref{thm:alon} the
  (undirected) pathwidth of~$D$ is bounded by~$O(k \log k)$. Thus, by Lemma
  \ref{lem:dp} the instance can be solved in time $O^*(2^{O(k\log^2 k)})$.
\end{proof}

\noindent
The following example demonstrates that Theorem \ref{thm:alon} does not hold
for arbitrary digraphs and thus the proof of Theorem \ref{thm:strong} cannot
be extended to the general case. Let $D$ be a digraph with vertex set
$\{v_0,v_1,\dots ,v_{n+1}\}$ and arc set $\{v_0v_1, v_1v_2, \ldots,
v_nv_{n+1}\}\cup \{v_iv_j:\ 1\le j<i\le n\}.$ Observe that $D$ is of unbounded
(undirected) treewidth, but has unique in- and out-branchings (which are
identical). The same statement holds if we add an arc $v_{n+1}v_0$ (to make
the graph strongly connected) but insist that the out-branching is rooted in
$v_0$ and the in-branching in $v_{n+1}$.

\section{The $k$-Distinct Branchings Problem}\label{sec:any}

\noindent
In this section, we fix a digraph~$D$ with terminals $s,t$ and simply talk
about \emph{rooted out-branchings (in-branchings)} whose root we implicitly
assume to be $s$ ($t$). Similarly, unless otherwise noted, a \emph{rooted 
out-tree (in-tree)} is understood to be rooted at~$s$ ($t$).

Clearly, to show that both versions of \textsc{$k$-Distinct Branchings} are FPT it is sufficient to prove the following:

\begin{theorem}\label{thm:main}
 \textsc{Rooted $k$-Distinct Branchings} is FPT for arbitrary digraphs.
\end{theorem}

\noindent In the rest of this section, $(D,s,t)$ will stand for an instance of
\textsc{Rooted $k$-Distinct Branchings} (in particular, $D$ is an input digraph of the problem) and $H$ for an
arbitrary digraph. Let us start by observing what further restrictions on $D$
can be imposed  by polynomial-time preprocessing.

\subsection{Preprocessing}\label{sec:preprocessing}
Let $(D,s,t)$ be an instance of \textsc{Rooted $k$-Distinct Branchings}. Recall
that $D$ contains an out-branching (in-branching) if and only if $D$ has only
one strongly connected component with no incoming (no outgoing) arc. As a
first preprocessing step, we can decide in polynomial time whether $D$ has a
rooted out-branching and a rooted in-branching. If not, we reject the
instance. Note that this in particular means that in a non-rejected instance,
every vertex in~$D$ is reachable from~$s$ and $t$ is reachable from every
vertex.

Next, we test for every arc~$a \in D$ whether there exists at least one rooted
in- or out-branching that uses~$a$ as follows: since a maximal-weight out- or
in-branching for an arc-weighted digraph can be computed in polynomial time
\cite{edmondsJRNBSS71}, we can force the arc~$a$ to be contained in a solution by
assigning it a weight of 2 and every other arc weight 1.
If we verify that $a$ indeed does not appears in any rooted out-branching and in-branching,
we remove~$a$ from~$D$ and obtain an equivalent instance of \textsc{Rooted $k$-Distinct Branchings}.

After this polynomial-time preprocessing, our instance has the following
three properties: there exists a rooted out-branching, there exists a rooted 
in-branching, and every arc of~$D$ appears in some rooted in- or out-branching.
We call such a digraph with a pair~$s,t$ \emph{reduced}.

Lastly, the following result of Kneis {\em et al.} \cite{kneisA61}
will be frequently used in our arguments below.

\begin{lemma}\label{lem:treeext}
Let $H=(V,A)$ be a digraph containing an out-branching rooted at $s\in V$. Then every out-tree rooted at $s$ with $q$ leaves  can be extended into an out-branching rooted at $s$ with at least $q$ leaves
in time $O(|V|+|A|)$.
\end{lemma}

\subsection{Decomposition and Reconfiguration}

\noindent We work towards the following win-win scenario: either we find an out-tree
with~$\Theta(k)$ leaves that can be turned into a rooted out-tree with at
least~$k+1$ leaves, or we conclude that every out-tree in~$D$ has less
than~$\Theta(k)$ leaves. We refer to the process of turning an out-tree into a
rooted out-tree as a \emph{reconfiguration}. In the process we will develop a
new digraph decomposition, the \emph{rooted cut-decomposition}, which will aid
us in reasoning about reconfiguration steps and ultimately lead us to a
solution for the problem. In principle we recursively
decompose the digraph into vertex sets that are bi-reachable from a designated
`bottleneck' vertex, but for technical reasons the  following  notion of
a \emph{diblock} results in a much cleaner version of the decomposition.

\begin{definition}
  Let $H$ be a digraph  with at least two vertices, and let $r \in V(H)$ such that every vertex of $H$ is reachable from $r$.
  Let $B \subseteq V(H)$ be the set of all vertices that are bi-reachable from $r$. 
  The \emph{directed block (diblock) $B_r$ of $r$ in $H$} is the set $B \cup N^+[r]$, i.e.,
  the bi-reachable vertices together with all out-neighbors of $r$ and $r$ itself. 
\end{definition}

\noindent 
Note that according to the above definition a diblock must have at least two vertices.

The following statement provides us with an easy case in which a
reconfiguration is successful, that is, we can turn an arbitrary 
out-tree into a rooted out-tree without losing too many leaves. 
Later, the obstructions to this case will be turned into
building blocks of the decomposition.

\begin{lemma}\label{lemma:re-root}
  Let~$B_s \subseteq V(D)$ be the diblock of~$s$ and
  let~$T$ be an out-tree of~$D$ whose root~$r$ lies in~$B_s$ with $\ell$
  leaves. Then there exists a rooted out-tree with at
  least~$(\ell-1)/2$ leaves.
\end{lemma}
\begin{proof}
 We may assume that $r\neq s$. In case~$T$ contains~$s$ as a leaf, we remove~$s$ from~$T$ for
  the remaining argument and hence will argue about the~$\ell-1$ remaining leaves.

  If $r$ is bi-reachable from $s$, consider two internally vertex-disjoint
  paths ~$P, Q$ from~$s$ to~$r$. One of the two paths necessarily avoids half
  of the~$\ell-1$ leaves of $ T$; let without loss of generality this path
  be~$P$. Let further~$L$ be the set of those leaves of~$ T$ that do
  \emph{not} lie on~$P$. If $r\in N^+(s)$, let $P=sr$.
  
  We construct the required out-tree $T'$ as follows:  first, add all arcs and
  vertices of~$P$ to~$T'$. Now for every leaf~$v \in L$, let $P_v$ be  the
  unique path   from~$r$ to~$v$ in~$ T$ and let $P'_v$ be the segment of $P_v$
  from the last vertex $x$ of $P_v$ contained in $T$. Add all arcs and
  vertices of~$P'_v$ to~$T'$. Observe that $x\neq v$ as $v$ cannot be in $T'$. Since $P_v$ and thus $P'_v$ contains no leaf of $L$ other than $v$,
  in the end of the process, all vertices of $L$ are leaves of $T'$. Since
  $|L| \geq (\ell-1)/2$, the claim follows.
\end{proof}

\noindent The definition of diblocks can also be understood in terms of network flows: Let $v\neq r.$
Consider the vertex-capacitated version of $H$ where $r$ and $v$ both have capacity 2, and every other vertex has capacity 1, for some $v \in V(H) \setminus \{r\}$.
Then $v$ is contained in the diblock of $r$ in $H$ if and only if the max-flow from $r$ to $v$ equals 2. 
Dually, by Menger's theorem, $v$ is \emph{not} contained in the diblock if and only if there is a 
vertex $u \notin \{r,v\}$ such that all $r$-$v$ paths $P$ intersect $u$. 
This has the following simple consequence regarding connectivity inside a diblock: \looseness-1

\begin{lemma}\label{lemma:bi-reachable-pair}
  Fix~$r \in V(H)$ and let~$B_r \subseteq V(H)$ be the diblock of $r$ in $H$. 
  Then for every pair of distinct vertices
  $x,y \in B_r$, there exist an $r$-$x$-path~$P_x$ and an $r$-$y$-path~$P_y$ that intersect
  only in~$r$.
\end{lemma}
\begin{proof}
  If $r \in \{x,y\}$, then clearly the claim holds since every vertex in $B_r$ is reachable from $r$.
  Otherwise, add a new vertex $z$ with arcs $xz$ and $yz$, and note that the lemma holds if and only if
  $z$ is bi-reachable from $r$. If this is not true, then by Menger's theorem there is a vertex
  $v \in B_r$, $v \neq r$, such that all paths from $r$ to $z$, and hence to $x$ and $y$, go through $v$.
  But as noted above, there is no cut-vertex $v \notin \{x,r\}$ for $r$-$x$ paths, 
  and no cut-vertex $v \notin \{y,r\}$ for $r$-$y$ paths.
  We conclude that $z$ is bi-reachable from $r$, hence the lemma holds.
\end{proof}

\noindent Next, we will use Lemma~\ref{lemma:bi-reachable-pair} to show that 
given a vertex $r$, the set of vertices not in the diblock $B_r$ of $r$ in $H$
partitions cleanly around $B_r$.

\begin{lemma}\label{lemma:partition-bottlenecks}
  Let $r \in V(H)$ be given, such that every vertex of $H$ is reachable from $r$.
  Let $B_r \subset V(H)$ be the diblock of $r$ in $H$. 
  Then $V(H) \setminus B_r$ partitions according to cut vertices in $B_r$, 
  in the following sense: For every $v \in V(H) \setminus B_r$, there is
  a unique vertex $x \in B_r \setminus \{r\}$ such that 
  every path from $r$ to $v$ intersects $B_r$ for the last time in $x$. 
  Furthermore, this partition can be computed in polynomial time.
\end{lemma}
\begin{proof}
  Assume towards a contradiction that for~$v \in V(H) \setminus B_r$ there exist
  two $r$-$v$-paths~$P_1, P_2$ that intersect~$B_r$ for the last
  time in distinct vertices~$x_1, x_2$, respectively. We first observe that $r \notin \{x_1,x_2\}$,
  since the second vertices of $P_1$ and $P_2$ are contained in $B_r$ by definition. 
  By Lemma~\ref{lemma:bi-reachable-pair},
  we may assume that $P_1[r,x_1] \cap P_2[r,x_2] = \{r\}$. But then~$P_1$ and~$P_2$ intersect for the first time
  outside of~$B_r$ in some vertex~$v'$ (potentially in~$v' = v$). This
  vertex is, however, bi-reachable from~$r$, contradicting our construction of~$B_r$.
  Hence there is a vertex $x \in B_r$ such that every path from $r$ to $v$
  intersects $B_r$ for the last time in $x$, with $x\neq r$, and clearly this vertex is unique.
  Finally, the set $B_r$ can be computed in polynomial time,
  and given $B_r$ it is easy to compute for each $x \in B_r$ the set of all vertices
  $v \in V(H)$ (if any) for which $x$ is a cut vertex.  
\end{proof}

\noindent We refer to the vertices $x \in B_r$ that are cut vertices in the above partition
as the \emph{bottlenecks of $B_r$}. Note that $r$ itself is not considered a bottleneck in $B_r$.
Using these notions, we can now define a \emph{cut decomposition}
of a digraph $H$.

\begin{definition}[Rooted cut decomposition and its tree]\label{def:decomp}
  Let~$H$ be a digraph and~$r$ a vertex such that every vertex in~$H$ is
  reachable from~$r$. The \emph{($r$-rooted) cut decomposition} of $H$ 
  is a pair~$(\hat T,\mathcal{B})$ where $\hat T$ is a rooted tree with $V(\hat T) \subseteq V(H)$
  and $\mathcal{B}=\{B_x\}_{x \in \hat T}$, $B_x \subseteq V(H)$ for each $x \in \hat T$, 
  is a collection of diblocks associated with the nodes of $\hat T$,  
  defined and computed recursively as follows.
  \begin{enumerate}
  \item Let $B_r$ be the diblock of $r$ in $H$, and let $L \! \subseteq \! B_r \! \setminus \! \{r\}$ 
    be the set of bottlenecks in $B_r$. Let $\{X_x\}_{x \in L}$ be the corresponding
    partition of the remainder $V(H) \! \setminus \! B_r$.
  \item For every bottleneck $x \in L$, let $(\hat T_x, \mathcal{B}_x)$ be the $x$-rooted cut decomposition
    of the subgraph $D[X_x \cup \{x\}]$.
  \item $\hat T$ is the tree with root node $r$, where $L$ is the set of children of $r$,
    and for every $x \in L$ the subtree of $\hat T$ rooted at $x$ is $\hat{T}_x$. 
  \item Finally, $\mathcal{B}=\{B_r\} \cup \bigcup_{x \in L} \mathcal{B}_x$.    
  \end{enumerate}
  Furthermore, for every node $x \in \hat T$, we define $B_x^*=\bigcup_{y \in \hat{T}_x} B_y$ as the set of all vertices contained in diblocks associated with nodes of the subtree $\hat{T}_x$. 
\end{definition}

\noindent Figure~\ref{fig:decomp} provides an illustration to Definition~\ref{def:decomp}.
\begin{figure}
  \centering\includegraphics[scale=.7]{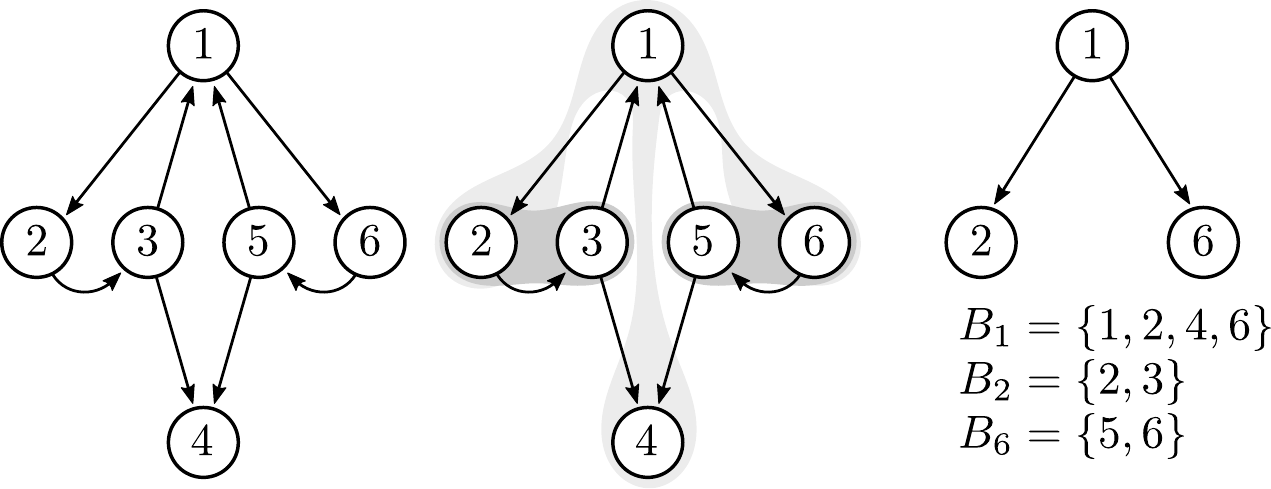}
  \caption{\label{fig:decomp}%
   An example of a rooted cut decomposition.}\vspace*{-.5em}
\end{figure}

\begin{lemma} \label{lemma:decomposition-well-defined}
  Let a digraph $H$ and a root $r \in V(H)$ be given, such that
  every vertex of $H$ is reachable from $r$. 
  Then the $r$-rooted cut decomposition $(\hat T, \{B_x\}_{x \in \hat{T}})$ of $H$ 
  is well-defined and can be computed in polynomial time.
  Furthermore, the diblocks cover $V(H)$, i.e., $\bigcup_{x \in \hat T} B_x=V(H)$,
  and for every node $x \in \hat T$,
    every vertex of $B_x^*$ is reachable from $x$ in $D[B_x^*]$.\looseness-1
\end{lemma}
\begin{proof}
  By Lemma~\ref{lemma:partition-bottlenecks}, the root diblock $B_r$
  as well as the set $L \subseteq B_r$ of bottlenecks and the partition $\{X_x\}_{x \in L}$
  are well-defined and can be computed in polynomial time. 
  Also note that for each $x \in L$, $r \notin X_x \cup \{x\}$, 
  and every vertex of $H_x:=H[X_x \cup \{x\}]$ is reachable from $x$ in $H_x$ by the definition of the partition.
  Hence the collection of recursive calls made in the construction is well-defined,
  and every digraph $H_x$ used in a recursive call is smaller than $H$, hence the process terminates.
  Finally, for any two distinct bottlenecks $x, y \in L$ we have $V(H_x) \cap V(H_y) = \emptyset$.  
  Thereby, distinct nodes of $\hat T$ are associated with distinct vertices of $H$,
  $|\hat T| \leq |V(H)|$, and the map $x \mapsto B_x$ is well-defined.
  It is also clear that the whole process takes polynomial time.  
\end{proof}

\noindent
We collect some basic facts about cut decompositions.

\begin{lemma}\label{lemma:cut-decomposition-props}
  Let~$H$ be a digraph, $r \in V(H)$ a vertex
  and let $(\hat T,\{B_x\}_{x \in \hat T})$ be the $r$-rooted cut decomposition of $H$. 
  Then the following hold.
  \begin{enumerate}
  \item The sets $\{B_x \setminus \{x\}\}_{x \in \hat T}$ 
    are all non-empty and partition $V(H)\setminus \{r\}$.
  \item For distinct nodes $x, y \in \hat T$, 
    if $x$ is the parent of $y$ in $\hat T$ then $B_x\cap B_y=\{y\}$;
    in every other situation, $B_x \cap B_y=\emptyset$. 
  \item For every node $x \in \hat T$, the following hold:
    \label{item:cut-prop-3}
    \begin{enumerate}
    \item If $y$ is a child of $x$ in $\hat T$, 
      then any arc leading into the set $B_y^*$ from $V(H) \setminus B_y^*$
      will have the form $uy$ where $u \in B_x$.
    \item If $y$, $y'$ are distinct children of $x$ in $\hat T$,
      then there is no arc between $B_y^*$ and $B_{y'}^*$.
    \end{enumerate}
  \end{enumerate}
  In particular, every arc of $H$ is either contained in a subgraph of $H$ induced by 
  a diblock $B_x$, or it is a \emph{back arc}
  going from a diblock $B_y$ to a diblock $B_x$, where $x$
  is an ancestor of $y$ in $\hat T$.
\end{lemma}
\begin{proof}
  For the first claim, the sets $B_x \setminus \{x\}$ are non-empty by definition; we show the partitioning claim.
  By Lemma~\ref{lemma:partition-bottlenecks},
  for every $v \in V(H) \setminus \{r\}$ either $v \in B_r \setminus \{r\}$
  or there is exactly one bottleneck $x \in B_r$ such that
  $v \in X_x$ in the construction of the decomposition.
  Also note that in the latter case, $v \neq x$ since $x \in B_r$.
  Applying the argument recursively and using that the diblocks
  cover $V(H)$, by Lemma~\ref{lemma:decomposition-well-defined}, 
 we complete the proof of the partitioning claim.

  For the second claim, the partitioning claim implies that
  if $v \in B_x \cap B_y$ for distinct nodes $x, y \in \hat T$,
  then either $v=x$ or $v=y$, i.e., $v$ must be a bottleneck.
  This is only possible in the situation described.

  For Claim 3(b), first consider the diblock $B_r$and the partition
  $\{X_z\}_{z\in L}$ given by Lemma~\ref{lemma:partition-bottlenecks}. To
  prove Claim 3(b) it suffices to show that for any two distinct sets $X_y$,
  $X_{y'}$ of the partition, there is no arc between $X_y$ and $X_{y'}$.
  Suppose for a contradiction that there is such an arc $uv$, $u \in X_y$, $v
  \in X_{y'}$. By Lemma~\ref{lemma:bi-reachable-pair}, there are paths $P_y$
  and $P_{y'}$ in $B_r$ from $r$ to $y$ and $y'$, respectively that intersect
  only in $r$, and by Lemma~\ref{lemma:decomposition-well-defined}, there are
  paths $P_u$ from $y$ to $u$ in $X_y$ and $P_v$ from $y'$ to $v$ in $X_{y'}$.
  But then the paths $P_yP_uuv$ and $P_{y'}P_v$ form two $r$-$v$ paths that
  are internally vertex-disjoint, showing that $v \in B_r$, contrary to our
  assumptions. Since the decomposition is computed recursively, this also
  holds in every internal node of $\hat T$.

  For Claim 3(a), let $uv$ be an arc such that $u\notin B^*_y$ and $v\in B^*_y$.
  Moreover, let $u\in B_{x'}$ and $v\in B_{y'}$. By construction of cut
  decomposition, there is a path $\hat{P}$ from $x'$ to $y'$ in $\hat T$
  containing nodes $x$ and $y$. Let $x''$ be the second node in $\hat P$ (just
  after $x'$). Thus, there is a path $P$ from $x''$ to $v$ in $H$ containing the
  vertices of $\hat P$ apart from $x'$.

  Assume that $u\neq x''$. Then by Lemma~\ref{lemma:bi-reachable-pair}, there is
  an $x'$-$u$-path $P'$ and an $x'$-$x''$-path $P''$ of $H$ which intersect only
  at $x'$. Then $x'P'uv$ and $P''P$ are internally vertex-disjoint paths from
  $x'$ to $v$. This means that $v$ must be in $B_{x'}$, a contradiction unless
  $x'=x$, $u\in B_x$ and $v=y$. If $u= x''$, then $P$ and $uv$ are internally
  vertex-disjoint paths from $u$ to $v$. This means that $v$ must be in
  $B_{x''}$, a contradiction unless $x'=x$ and $v=y$. 
\end{proof}
 
\noindent As we saw, for every diblock $B_y$, $y \in \hat T$, 
any path ``into'' the diblock must go via the bottleneck vertex $y$. 
By induction, for any $v \in B_y$, every node of $\hat T$ from $r$ to $y$ 
represents a bottleneck vertex that is unavoidable for paths from $r$ to $v$.
More formally, the following holds in cut decompositions:

\begin{lemma}\label{lemma:ancestor-bottlenecks}
  Let~$(\hat T,\{B_x\}_{x\in \hat T})$ be the cut decomposition of $H$ rooted at~$r$. Fix a 
  diblock~$B_x$ for~$x \in \hat T$. Consider a path~$P$ in $H$ from~$r$ to~$v \in B_x$
  and let~$x_1 \ldots x_\ell$ be the sequence of bottleneck vertices that~$P$ encounters.
  Then~$\hat P = x_0x_1\ldots x_\ell$ with~$x_0 = r$ is the path from~$r$ to~$x$
  in~$\hat T$. 
\end{lemma}
\begin{proof}
  We prove the claim by induction over the depth~$d$ of the vertex~$x$ in~$\hat T$. If~$r = x$
  then any path from~$r$ to~$v \in B_r$ contains~$r$ itself and hence the
  base case for~$d = 0$ holds trivially.

  Consider a diblock~$B_x$, $x \in \hat T$ where~$x$ has distance~$d$ to~$r$ in~$\hat T$
  and let~$y$ be the parent of~$x$ in~$\hat T$.
  We assume the induction hypothesis holds for diblocks at depth~$d-1$, hence
  it holds for~$B_y$ in particular. Because~$x \in B_y$, this implies that
  every path from~$r$ to~$x$ will contain all ancestors of~$x$ in~$\hat T$. 
  Since by construction every path from~$r$ to a vertex~$v \in B_x$ needs to pass 
  through~$x$, the inductive step holds. This proves the claim.
\end{proof}

\noindent
As an immediate consequence, we can identify arcs in cut decompositions that
cannot participate in any rooted out-branching.

\begin{corollary}\label{cor:marked-out-arcs}
  Let~$(\hat T,\{B_x\}_{x\in \hat T})$ be the cut decomposition of $H$ rooted at~$r$ and let
  $R := \{ uv \in A(H) \mid u \in B_x~\text{and}~x \in \hat T_v \}$ be all the
  arcs that originate in a diblock~$B_x$ and end in an ancestor~$v$ of~$x$ on $\hat T$.
  Then for every  out-tree~$T$ rooted at $r$ we have~$A(T) \cap R = \emptyset$.
\end{corollary}
\begin{proof}
  Fix a bottleneck vertex~$v \in \hat T$ of the decomposition and let the arc~$uv$
  be in an out-tree~$T$ rooted at $r$. There must exist a path~$P_{su}$
  from~$s$ to~$u$ that is part of~$T$. By Lemma~\ref{lemma:ancestor-bottlenecks},
  this path will contain the vertex~$v$. But then~$v$ is an ancestor of~$u$
  in~$T$ and therefore the arc~$uv$ cannot be part of~$T$, which is a contradiction.
\end{proof}

\noindent
The decomposition actually restricts paths even further: a path that starts at
the root and visits two bottleneck vertices~$x,y$ (in this order) cannot
intersect any vertex of $B_y^*$ before visiting $y$ and cannot return to any
set $B_z^*$, $z \in \hat T$, after having left it.

\begin{lemma}\label{lemma:between-bottlenecks}
  Let~$(\hat T,\{B_x\}_{x\in \hat T})$ be the cut decomposition of $H$ rooted at~$r$. Fix a 
  diblock~$B_x$ for~$x \in \hat T$. Consider a path~$P$ from~$r$ to~$v \in B_x$
  and let~$\hat P = x_0 \ldots x_\ell$ be the path from~$r = x_0$ to~$x = x_\ell$
  in~$\hat T$. Let further~$F_0, \ldots, F_\ell$ be the fins of~$\hat P$ in~$\hat T$.
  Then the subpath $P[x_i,x_{i+1}] \setminus \{x_{i+1}\}$ is contained in the
  union of diblocks of $F_i$ for~$0 \leq i < \ell$.
\end{lemma}
\begin{proof}
  By Lemma~\ref{lemma:ancestor-bottlenecks} we know that the nodes of~$\hat P$
  appear in~$P$ in the correct order, hence the subpath~$P[x_i,x_{i+1}]$ is well-defined.
  Let us first show that the subpath~$P[x_i,x_{i+1}]\setminus \{x_{i+1}\}$ cannot 
  intersect any diblock associated with~$\hat T_{x_{i+1}}$. By
  Lemma~\ref{lemma:cut-decomposition-props}, the only arcs from~$B_{x_i}$ into
  diblocks below~$x_{i+1}$ connect to the bottleneck~$x_{i+1}$ itself.
  Since~$x_{i+1}$ is already the endpoint of~$P[x_i,x_{i+1}]$, this
  subpath cannot intersect the diblocks of~$\hat T_{x_{i+1}}$. This already proves the
  claim for~$x_0$; it remains to show that it does not intersect diblocks of~$V(\hat T)
  \setminus V(\hat T_{x_i})$ for $i \geq 1$. The reason is similar: since the
  bottleneck~$x_i$ is already part of~$P[x_i,x_{i+1}]$, this subpath
  could not revisit~$B_{x_i}$ if it enters any diblock~$B_y$ for a proper
  ancestor~$y$ of~$x_i$ in $\hat T$. We conclude that therefore it must be, with the
  exception of the vertex~$x_{i+1}$, inside the diblocks of the fin~$F_i$.
\end{proof}

\begin{corollary}\label{cor:avoid-half}
  For every vertex~$u \in V(H)$ and every set~$X \subseteq V(H)\setminus (V(\hat T)
  \cup \{u\})$ of non-bottleneck vertices there exists a path~$P$
  from~$r$ to~$u$ in~$H$ such that~$|P \cap X| \leq |X|/2$.
\end{corollary}
\begin{proof}
  Assume that~$u \in B_x$ and let~$\hat P = x_0\ldots x_\ell$ be a path from~$x_0 = r$
  to~$x_\ell = x$ in~$\hat T$. Let further~$F_0,\ldots,F_\ell$ be the fins
  of~$\hat P$ in~$\hat T$  and $U_i$ the union of diblocks associated with~$F_i$, $0 \leq i \leq \ell$. 
  We partition the set~$X$ into~$X_1,\ldots,X_\ell$
  where~$X_i = X \cap U_i$ for~$0 \leq i \leq \ell$. 
Lemma~\ref{lemma:between-bottlenecks} allows us to construct
  the path~$P$ iteratively: any path that leads to~$u$ will
  pass through bottlenecks~$x_i,x_{i+1}$ in succession and visit
  only vertices in $U_i$ in the process. 
  Since there
  are two internally vertex-disjoint paths between~$x_i,x_{i+1}$
  for~$1 \leq i \leq \ell$, we can always choose the path that
  has the smaller intersection with~$X_i$. Stringing these paths
  together, we obtain the claimed path~$P$.
\end{proof}

\noindent
We want to argue that one of the following cases must hold: either the
cut decomposition has bounded height and we can `re-root' any out-tree with 
many leaves into a rooted out-tree with a comparable number of leaves,
or we can directly construct a rooted out-tree with many
leaves. In both cases we apply Lemmas~\ref{lemma:leaves-branch} and~\ref{lem:treeext} to conclude that the instance
has a solution. This approach has one obstacle: internal diblocks of the decomposition
that contain only two vertices.

\begin{definition}[Degenerate diblocks]
    Let~$\{B_x\}_{x \in \hat T}$ be the cut decomposition rooted at~$s$. We call a
    diblock~$B_x$ \emph{degenerate} if~$x$ is an internal node of~$\hat T$
    and~$|B_x| = 2$.
\end{definition}

\noindent
Let us first convince ourselves that a long enough sequence of non-degenerate
diblocks provides us with a rooted out-branching with many leaves.

\begin{lemma}\label{lemma:win-non-degen-path}
  Let~$(\hat T,\{B_x\}_{x\in \hat T})$ be the cut decomposition rooted at~$s$ of $H$ and let
  $y$ be a node in~$\hat T$ such that the path~$\hat P_{sy}$ from~$s$ to~$y$ in~$\hat T$
  contains at least~$\ell$ nodes whose diblocks are non-degenerate.
  Then~$H$ contains an out-tree rooted at $s$ with at least~$\ell$ leaves.
\end{lemma}
\begin{proof}
  We construct an $s$-rooted out-tree $T$ by repeated application of Lemma~\ref{lemma:bi-reachable-pair}. 
  Let $x_1$, $\ldots$, $x_\ell$ be a sequence of nodes in $\hat P_{sy}$ whose diblocks are non-degenerate,
  and for each $1 \leq i < \ell$ let $x^+_i$ be the node after $x_i$ in $\hat P_{sy}$.
  We construct a sequence of $s$-rooted out-trees $T_1$, $\ldots$, $T_\ell$ such that 
  for $1 \leq i \leq \ell$, the vertex $x_i$ is a leaf of $T_i$, 
  and $T_i$ contains $i$ leaves. 
  First construct $T_1$ as a path from $s$ to $x_1$, then for every $1 \leq i <\ell$ 
  we construct an out-tree $T_{i+1}$ from $T_i$ as follows.
  Let $v_i \in B_{x_i} \setminus \{x_i,x^+_i\}$, which exists since $B_{x_i}$ is non-degenerate,
  and let $P_{x_ix^+_i}$, $P_{x_iv_i}$ be a pair of paths in $D[B^*_{x_i}]$ 
  from $x_i$ to $x^+_i$ and to $v_i$ respectively, which intersect only in $x_i$. 
  Such paths exist by Lemma~\ref{lemma:bi-reachable-pair}, and
  since $x_i$ is a leaf of $T_i$, Lemma~\ref{lemma:ancestor-bottlenecks}
  implies that $T_i$ is disjoint from $B_{x_i}^* \setminus \{x_i\}$.
  Hence the paths can be appended to $T_i$ to form a new $r$-rooted out-tree 
  $T_{i+1}$ in $H$ which contains a leaf in every diblock $B_{x_j}$, $1 \leq i$. 
  Finally, note that the final tree $T_\ell$ contains two leaves in $B_{x_{\ell-1}}$,
  hence $T_\ell$ is an $r$-rooted out-tree with $\ell$ leaves. 
\end{proof}

\noindent  
The next lemma is important to prove that \textsc{Rooted $k$-Distinct
Branchings} is FPT for a special case of the problem considered in Lemma
\ref{thm:boundedd}.

\begin{lemma}\label{lemma:win-bounded-height}
   Let~$(\hat T,\{B_x\}_{x\in \hat T})$ be the cut decomposition of $D$ rooted at~$s$ such
   that $\hat T$ is of height~$d$  and let $T$ be an out-tree rooted at some
   vertex~$r$ with~$\ell$ leaves. Then we can construct an out-tree~$T_s$
   rooted at~$s$ with at least~$(\ell - d)/ 2$ leaves. \looseness-1
\end{lemma}
\begin{proof}
  Assume that~$r$ is contained in the diblock~$B_x$ of the decomposition and
  let~$x_p \ldots x_1 = \hat P_{sx}$ be a path from~$s = x_p$ to~$x = x_1$ in~$\hat T$.
  Let~$L$ be the leaves of~$T$ and let~$L' := L \setminus \hat P_{sx}$. 
  Clearly, $|L'| \geq \ell-d$. Applying Corollary~\ref{cor:avoid-half}
  with~$X = L'$ and~$u = r$, we obtain a path~$P_{sr}$ in $D$ from $s$ to $r$ that avoids
  half of~$L'$. We construct~$T_s$  in a similar fashion to the proof of Lemma~\ref{lemma:re-root}.
  We begin with~$T_s = P_{sr}$, then 
  for every leaf~$v \in L'\setminus P_{sr}$, proceed as follows: 
  let $P_v$ be the unique path from~$r$ to~$v$ in~$T$ and 
  let $P'_v$ be the segment of $P_v$ from the last vertex $x$ of $P_v$ contained in $T_s$. 
  Add all arcs and vertices of~$P'_v$ to~$T_s$. Since $P_v$ and thus $P'_v$
  contains no leaf of $L'$ other than $v$, in the end of the process, all
  vertices of $L'\setminus P_{sr}$ are leaves of $T_s$. Since
  $|L'\setminus P_{sr}| \geq |L'|/2$, we conclude that~$T_s$ contains
  at least~$(\ell - d)/2$ leaves, as claimed.
\end{proof}

\noindent 
Using these results, we are now able to prove that if the height
$d$ of the cut decomposition of $D$ is upper-bounded by a function in
$k$, then \textsc{Rooted $k$-Distinct Branchings} on $D$ is FPT.  
Combined with Lemma~\ref{lemma:win-non-degen-path},  this implies
that the remaining obstacle is cut decompositions with long chains of  
degenerate diblocks, which we will deal with
in Section~\ref{sec:degenerate}.

\begin{lemma}\label{thm:boundedd}
  Let $(\hat T,\{B_x\}_{x\in \hat T})$ the cut decomposition rooted at $s$ of
  height $d$. If $d\le d(k)$ for some function $d(k)=\Omega(k)$ of $k$ only,
  then we can solve \textsc{Rooted $k$-Distinct Branchings} on $D$ in time
  $O^*(2^{O(d(k)\log^2 d(k))})$.
\end{lemma}
\begin{proof} 
  By Theorem~\ref{thm:kleaf}, in time $O^*(2^{O(d(k))})$ we can decide whether
  $D$ has an out-branching with at least $2k+2+d(k)$ leaves. If $D$ has such an
  out-branching, then by Lemma~\ref{lemma:win-bounded-height} $D$ has a rooted
  out-tree with at least $k+1$ leaves. This out-tree can be extended to a rooted
  out-branching with at least $k+1$ leaves by Lemma~\ref{lem:treeext}. So by
  Lemma~\ref{lemma:leaves-branch}, $(D,s,t)$ is a positive instance if and only
  if $D$ has a rooted in-branching, which can be decided in polynomial time.

  If $D$ has no out-branching with at least $2k+2+d(k)$ leaves, by
  Theorem~\ref{thm:alon} the pathwidth of $D$ is $O(d(k)\log d(k))$ and thus by
  Lemma~\ref{lem:dp} we can solve  \textsc{Rooted $k$-Distinct Branchings} on
  $D$ in time $O^*(2^{O(d(k)\log^2 d(k))})$. (Note that for the dynamic
  programming algorithm of Lemma~\ref{lem:dp} we may fix roots of all out-
  branchings and all in-branchings of $D$ by adding arcs $s's$ and $tt'$ to $D$,
  where $s'$ and $t'$ are new vertices.)
\end{proof}

\subsection{Handling degenerate diblocks}\label{sec:degenerate}

The following is the key notion for our study of degenerate diblocks. 

\begin{definition}[Degenerate paths] Let~$(\hat T,\{B_x\}_{x\in \hat T})$ be a cut decomposition of $D$.  We call a path~$\hat P$ in~$\hat T$ \emph{monotone} if it is a subpath of a path from the root of~$\hat T$ to some leaf of~$\hat T$.
  We call a path~$\hat P$ in~$\hat T$
  \emph{degenerate} if it is monotone and every diblock~$B_x$, $x \in \hat P$ is degenerate.
\end{definition}

\noindent Let $(D,s,t)$ be a reduced instance of \textsc{Rooted $k$-Distinct Branchings}.
As observed in Section~\ref{sec:preprocessing}, we can verify in polynomial time whether an
arc participates in \emph{some} rooted in- or out-branching. Let~$R_s \subseteq A(D)$
be those arcs that do not participate in any rooted out-branching and $R_t \subseteq A(D)$
those that do not participate in any rooted in-branching. Since $(D,s,t)$ is a reduced instance, we necessarily
have that~$R_s \cap R_t = \emptyset$, a fact we will use frequently in the following.
Corollary~\ref{cor:marked-out-arcs} provides us with an important subset of~$R_s$: every
arc that originates in a diblock~$B_x$ of the cut decomposition and ends in a bottleneck vertex
that is an ancestor of~$x$ on $\hat T$ is contained in~$R_s$.

Let us first prove some basic properties of degenerate paths.

\begin{lemma}\label{lemma:degen-path-props}
  Let~$(\hat T,\{B_x\}_{x\in \hat T})$ be the cut-decomposition of $D$ rooted at~$s$,
  and let~$\hat P = x_1\ldots x_\ell$ be a degenerate path of~$\hat T$.
  Then the following properties hold:
  \begin{enumerate}
    \item Every rooted out-branching contains~$A(\hat P)$,
    \item every arc~$x_jx_i$ with~$j > i$ is contained in~$R_s$, and
    \item there is no arc from~$x_i$ $(i<\ell)$ to~$B_y$ in $D$, 
          where~$y$ is a descendant of~$x_i$ on $\hat T$, except for the arc $x_ix_{i+1}$.
  \end{enumerate}
\end{lemma}
\begin{proof}
  First observe that, by definition of a degenerate path, $\hat
  P$ is a path in $D$. Every rooted out-branching contains in particular  the
  last vertex~$x_\ell$ of the path. By Lemma~\ref{lemma:ancestor-bottlenecks},
  it follows that $\hat P$ is contained in the out-branching as a monotone
  path, hence it contains~$A(\hat P)$. Consequently, no `back-arc' $x_jx_i$
  with~$j > i$ can be part of a rooted out-branching and thus it is contained
  in~$R_s$. For the third property, note that all arcs from $x_i$ except back
  arcs are contained in $B_{x_i}$. Since $B_{x_i}$ is degenerate there can be
  only one such arc.
\end{proof}

\noindent
For the remainder of this section, let us fix a single degenerate path~$\hat P = x_1\ldots x_\ell$.
We categorize the arcs incident to~$\hat P$ as follows:
\begin{enumerate}
  \item Let~$A^+$ contain all `upward arcs' that originate in~$\hat P$ and
        end in some diblock~$B_y$ where~$y$ is an ancestor of~$x_1$,
  \item let~$A^0$ contain all `on-path arcs' $x_jx_i$, $j > i$, and
  \item let~$A^-$ contain all `arcs from below' that originate from some diblock~$B_y$
        where~$y$ is a (not necessarily proper) descendant of~$x_\ell$.
\end{enumerate}
By Lemma~\ref{lemma:degen-path-props}, this categorization is complete: no
other arcs can be incident to~$\hat P$ in a reduced instance. By the same lemma,
we immediately obtain that~$A^0,A^- \subseteq R_s$.
We will now apply certain reduction rules to~$(D,s,t)$ and prove in the
following that they are safe, with the goal of bounding the size of~$\hat P$
by a function of the parameter~$k$.

\vspace{2mm}
\noindent{\textbf{Reduction rule 1:}} If there are two arcs~$x_iu, x_ju \in A^+ \cap R_t$ with~$i < j$, remove the arc~$x_ju$.
\vspace{2mm}

\begin{lemma}
  Rule 1 is safe.
\end{lemma}
\begin{proof}
  Since~$(D,s,t)$ is reduced, the arcs~$x_iu$ and~$x_ju$
  cannot be in~$R_s$. Pick any rooted out-branching~$T$ that contains the arc~$x_ju$.
  By Lemma~\ref{lemma:degen-path-props}, we have that~$\hat P \subseteq T$, therefore
  we can construct an out-branching~$T'$ by exchanging the arc~$x_ju$ for
  the arc~$x_iu$. Since a) no rooted in-branching contains either of these two arcs, and
  b) no out-branching can contain both,
   we conclude that~$(D\setminus\{x_ju\},s,t)$ is equivalent to~$(D,s,t)$
  and thus Rule 1 is safe.
\end{proof}

\begin{corollary}\label{cor:few-useless-up-arcs}
  Let~$(D,s,t)$ be reduced with respect to Rule 1. Then we either find
  a solution for~$(D,s,t)$ or $|A^+ \cap R_t| \leq 2k+1$.
\end{corollary}
\begin{proof}
  Let~$H$ be the heads of the arcs~$A^+ \cap R_t$. Since Rule~1
  was applied exhaustively, no vertex in~$H$ is the head of
  two arcs~$A^+ \cap R_t$; therefore we have that~$|H| = |A^+ \cap R_t|$.

  Note that any arc in~$A^+ \cap R_t$ cannot be contained in
  $R_s$, therefore $H$ does not contain any bottleneck vertices.
  Applying Corollary~\ref{cor:avoid-half}, we can find a path~$P_\ell$
  from~$s$ to~$x_\ell$ that avoids half of the vertices in~$H$. 
  Thus we can add half of~$H$ as leaves to~$P_\ell$ using the 
  arcs from~$A^+ \cap R_t$. Thus if~$|H| \geq 2k+2$, we obtain
  a rooted out-tree with at least~$k+1$ leaves,
  which by Lemmas~\ref{lemma:leaves-branch}
  and~\ref{lem:treeext} imply that the
  original instance has a solution. We conclude that otherwise
  $|H| = |A^+ \cap R_t| \leq 2k+1$.
\end{proof} 

\begin{lemma}\label{lemma:A-plus-tails} Let~$\hat P = x_1\ldots x_\ell$ be a degenerate path.
  Assume~$t \not \in \hat P$ and that $(D,s,t)$ is reduced with respect to Rule 1. Let
  further~$X \subseteq V(\hat P)$ be those vertices of~$\hat P$ that are tails of the arcs
  in~$A^+$. We either find that~$(D,s,t)$ has a solution or that~$|X| \leq 3k+1$.
\end{lemma}
\begin{proof}
  For every vertex~$x_i \in X$ with an arc $x_iv \in A^+\setminus R_t$, we
  construct a path~$\tilde P_{x_it}$ from~$x_i$ to~$t$ that contains $x_iv$ as
  follows: since $x_iv \not \in R_t$, there exists a rooted in-branching
  $\tilde T$ that contains~$x_iv$. We let~$\tilde P_{x_it} \subseteq \tilde T$
  be the path from~$x_i$ to~$t$ in~$\tilde T$.

  \begin{claim}
    Each path~$\tilde P_{x_it}$ does not intersect vertices in diblocks
    of~$\hat T_{x_i}$.
  \end{claim}
  \noindent Since~$\tilde P_{x_it}$ leaves~$\hat P$ via the first arc~$x_iv$, it
  cannot use the arc~$x_ix_{i+1}$. Since this is the only arc that leads to
  vertices in diblocks of~$\hat T_{x_i}$, the claim follows.\qed \smallskip

  \noindent Let us relabel the just constructed paths to~$\tilde P_1,\ldots, \tilde P_\ell$
  such that they are sorted with respect to their start vertices on~$\hat P$.
  That is, for~$i < j$ the first vertex of~$\tilde P_i$ appears before the
  first vertex of~$\tilde P_j$ on~$\hat P$. We iteratively construct rooted
  in-trees $\tilde T_1,\ldots \tilde T_\ell$ with the invariant that a)
  $\tilde T_i$ has exactly~$i$ leaves and b) does not contain any vertex
  of~$\hat P$ below the starting vertex of $\tilde P_i$. Choosing~$\tilde T_1
  := \tilde P_1$ clearly fulfills this invariant. To construct~$\tilde T_i$
  from~$\tilde T_{i-1}$ for~$2 \leq i \leq \ell$, we simply follow the
  path~$\tilde P_i$ up to the first intersection with~$\tilde T_{i-1}$.
  Since~$t \in \tilde P_i \cap \tilde T_{i-1}$, such a vertex must eventually
  exist. By the above claim, $\tilde P_i$ does not contain any vertex below
  its starting vertex on~$\hat P$, thus both parts of the invariant remain
  true. \looseness-1

  We conclude that~$\tilde T_{\ell}$ is a rooted in-tree with $\ell$ leaves,
  where $\ell$ is the number of vertices in~$X$ that have at least one 
  upwards arc not contained in~$R_t$. For~$\ell \geq k+1$,
  Lemmas~\ref{lemma:leaves-branch} and~\ref{lem:treeext} imply that the
  original instance has a solution. Otherwise, $\ell \leq k$.
  By Corollary~\ref{cor:few-useless-up-arcs}, we may assume that~$|A^+\cap R_t| \leq 2k+1$. Taken both facts together, we conclude that either
  $
    |X| \leq 3k+1
  $,
  or we can construct a solution.
\end{proof}

\noindent
We will need the following.

\begin{lemma}\label{lemma:A-zero-tails} 
  Let~$\hat P = x_1\ldots x_\ell$ be a degenerate path and assume that~$t \not \in \hat P$. Let 
  further~$Y \subseteq V(\hat P)$ be those vertices of~$\hat P$ that are tails of the arcs
  in~$A^0$. We either find that~$(D,s,t)$ has a solution or we have~$|Y| < k$.
\end{lemma}
\begin{proof}
We will construct a rooted in-tree that contains $|Y|$ arcs from $A^0$. 
  Since no rooted out-branching contains any such arc, this will prove that the instance is positive provided $|Y|\ge k$. 
  Note in particular that we are not concerned
with the number of leaves of the resulting tree.

  First associate every vertex~$v \in Y$ with an arc~$vx_v \in A^0$, where we choose~$x_v$
  to be the vertex closest to~$v$. Let~$X^+$ be the heads of all arcs in~$A^+$
  and let~$X \subseteq \hat P$ be the tails of all arcs in~$A^+$.

  Let~$u \in Y$ be the vertex that appears first on~$\hat P$ among all vertices in~$Y$.
  Since~$ux_u \in R_s$,
  $ux_u$ cannot be contained in~$R_t$. Accordingly there exists a path~$P_u$
  from~$u$ to~$t$ that uses the arc~$ux_u$. Note that~$P_u$ leaves~$\hat P$ through
  an $A^+$-arc whose tail lies between~$x_u$ and~$u$ on~$\hat P$. 

  Note that the segment~$P_u[x_u,t]$, by
  our choice of~$u$, does not contain any vertex of~$Y$. We now construct the
  \emph{seed in-tree}~$T_0$ as follows. We begin with~$P_u$ and add
  the arc $u'x_{u'}$ for every vertex~$u' \in Y$ where~$x_{u'} \in P_u$. 
  Next, we add every vertex $v \in X^+$ to~$T_0$ by finding a path from~$v$ to~$t$ and attach this
  path up to its first intersection with~$T_0$. Since~$v$ lies above~$x_u$
  in the decomposition, this path cannot intersect any vertex in~$Y$. 

  We form an in-forest\footnote{An {\em in-forest} is a vertex-disjoint collection of in-trees.}~$\mathcal F_0$ 
  from the arcs of $T_0$ and all arcs~$vx_v$, $v \in
  Y$ that are not in~$T_0$. 
  Every in-tree $T \in \mathcal F_0$ has
  the following easily verifiable properties:
  \begin{enumerate}
    \item \label{inv:1} Its root is the highest vertex in the decomposition
          among all vertices~$V(T)$ (recall Lemma~\ref{lemma:cut-decomposition-props}),
    \item \label{inv:2} its root is either~$t$ or a vertex~$x_v$ with~$v \in Y$, and 
    \item \label{inv:3} every segment~$\hat P[x,y]$ of~$\hat P$ contained in~$T$ has
          no vertex of~$X$ or~$Y$ with the possible exception of $y \in Y$.
  \end{enumerate}
   We will maintain all three of these
  properties while constructing a sequence of in-forests~$\mathcal F_0 \subseteq \mathcal
  F_1 \subseteq \ldots$
  where each in-forest in the sequence will have less roots than its
  predecessor (here $\subseteq$ stands for $\mathcal F_{i-1}$ being a subgraph of $\mathcal F_{i}$). 
  We stop the process when the number of roots drop to one.

  The construction of~$\mathcal F_i$ from~$\mathcal F_{i-1}$ for~$i \geq 1$ works as
  follows. Let~$T \in \mathcal F_{i-1}$ be the in-tree with the \emph{lowest}
  root in the in-forest. By assumption, $\mathcal F_{i-1}$ has at least two roots,
  thus it cannot be~$t$ and therefore, by part~\ref{inv:2} of our invariant, is a vertex~$x_v$ with~$v \in Y$.
  Now from~$x_v$ onwards, we walk along the path~$\hat P$ until we encounter a vertex~$z$
  that is either 1) the tail of an arc~$A^+$ or 2) a vertex of the in-forest~$\mathcal F_{i-1}$.
  In the former case, we use the segment~$\hat P[x_u,z]$ to connect~$T$ to a
  vertex in~$X^+$ via the $A^+$-arc emanating from~$z$. Since all vertices in~$X^+$
  are part of the seed in-tree~$T_0$ this preserves all three parts of the invariant
  and~$\mathcal F_{i-1} \subseteq \mathcal F_i$.

  Consider the second case: we encounter a vertex~$z$ that is part of some
  tree~$T' \in \mathcal F_{i-1}$. Let us first eliminate a 
  degenerate case:
  \begin{claim}
    The trees~$T'$ and~$T$ are distinct.
  \end{claim}
  \noindent Assume towards a contradiction that~$T' = T$.
  In that case, there is no single vertex between~$x_v$ and~$z$
  that is either the tail of an~$A^+$ or $A^0$-arc. If~$z$ is such that
  $x_z = x_v$ (for example $z=v$), we already
  obtain a contradiction: the arc~$zx_z$ cannot possibly be part of
  any in-branching rooted at~$t$, contradicting the fact that it is not in~$R_t$.

  Otherwise, $z = x_w$ for some $w \in Y$. By assumption, $wz \in T$, thus
  there exists a path $P_{wx_v}$ which contains both the arc $wz$ as well
  as the arc $vx_v$. Therefore, the vertex $w$ must lie below $v$ on $P$.
  Furthermore, the path $P_{wx_v}$ cannot contain any vertex above $x_v$ by
  property~\ref{inv:1} of the invariant. It follows that the subpath $P[z,v]$ is 
  entirely contained in $T$---by property~\ref{inv:3} of the invariant, none of the
  vertices (except $v$) in this subpath can be in $X$ or $Y$. Since the
  same was true for the subpath $P[x_v,z]$, we conclude that the whole
  subpath $P[x_v,v]$ contains no vertex of $X$ or $Y$. This contradicts
  our assumption that $x_vv \not \in R_t$ and we conclude that $T$ and $T'$
  must be distinct. \qed \\

  \noindent
  By our choice of~$T$, the vertex~$z$ cannot be the root of~$T'$.
  Accordingly, we can merge~$T$ and~$T'$ by adding the path~$\hat P[x_v,z]$.
  This concludes our construction of~$\mathcal F_i$. Since the root of~$T'$
  lies above all vertices of~$T$, part 1) of the invariant remains true.
  We did not change a non-root to a root in this construction, thus part 2)
  remains true. The one segment of~$\hat P[x_v,z]$ we added to merge~$T$
  and~$T'$ did not, by construction, contain any vertex of~$Y$ or~$X$,
  with the exception of the last vertex~$z$, hence part 3) remains true. 
  Finally, we clearly have that~$\mathcal F_{i-1} \subseteq \mathcal F_i$
  and the latter contains one less root than the former.

  The process clearly terminates with some in-forest~$\mathcal F_p$ which contains
  a single in-tree~$\tilde T$. The root of this in-tree is necessarily~$t$.
  Note further that~$\mathcal F_0$ contained all~$|Y|$ arcs~$ux_u$, $u \in Y$;
  therefore~$\tilde T$ contains those arcs, too. Since all of these arcs are
  in~$R_s$, we arrive at the following: either~$|Y| < k$, as claimed, or
  we found a rooted in-tree that avoids at least~$k$ arcs with every rooted
  out-branching, in other words, as solution to~$(D,s,t)$.
\end{proof}

\noindent 
Taking Lemma~\ref{lemma:A-plus-tails} and Lemma~\ref{lemma:A-zero-tails} together,
we see that only~$O(k)$ vertices of a degenerate path are tails of arcs in~$A^+$
or~$A^0$. The following lemma now finally lets us deal with degenerate paths: we argue
that those parts of the path that contain none of these few `interesting' vertices can
be contracted.

\vspace{2mm}
\noindent{\textbf{Reduction rule 2:}}
  If~$\hat P[x,y] \subseteq \hat P$ is such that no vertex in~$\hat P[x,y]$ is a tail
  of arcs in~$A^+ \cup A^0$, contract~$\hat P[x,y]$ into a single vertex.
\vspace*{2mm}

\begin{lemma}
  Rule 2 is safe.
\end{lemma}
\begin{proof}
  Simply note that every vertex in~$\hat P[x,y]$ has exactly one outgoing arc.
  We already know that every arc of~$\hat P$ must be contained in every rooted
  out-branching, now we additionally have that all arcs in~$\hat P[x,y]$ are necessarily 
  contained in every rooted in-branching. We conclude that Rule~2 is safe.
\end{proof}

\noindent We summarize the result of applying Rules~1 and~2.

\begin{lemma}\label{lem:last}
  Let~$\hat P$ be a degenerate path in an instance reduced with respect to Rules~1 and~2.
  If $t\not\in \hat P$ then~$|\hat P| \leq 8k + 1$. Otherwise, $|\hat P| \leq 16k + 3.$
\end{lemma}
\begin{proof}
  We may assume that $t\not\in \hat P$ as otherwise we can partition $\hat P$
  into its part before $t$ and its part after $t$ and obtain $|\hat  P| \leq
  16k + 3$ from the bound on $|\hat P|$ when $t\not\in \hat P$. By
  Lemmas~\ref{lemma:A-plus-tails} and~\ref{lemma:A-zero-tails}, the number of
  vertices on~$\hat P$ that are tails of either~$A^+$ or $A^0$ is bounded by
  $4k$. Between each consecutive pair of such vertices, we can have at most one vertex
  that is not a tail of such an arc. We conclude that~$|\hat P| \leq 8k + 1$, as
  claimed.
\end{proof}

\noindent Now we can prove the main result of this paper. \\

\noindent {\bf Proof of Theorem~\ref{thm:main}}. 
Consider the longest monotone path $\hat P$ of $\hat T$.
By Lemma~\ref{lemma:win-non-degen-path}, if $\hat P$ has at least $k+1$
non-degenerate diblocks, then $D$ has a rooted out-tree with at least $k+1$
leaves. This out-tree can be extended to a rooted out-branching with at least
$k+1$ leaves by Lemma~\ref{lem:treeext}. Thus, by Lemma~\ref{lemma:leaves-branch},
$(D,s,t)$ is a positive instance if and only if $D$ has a rooted in-branching, which can be decided in polynomial time.

Now assume that $\hat P$ has at most $k$ non-degenerate diblocks.  By
Lemma~\ref{lem:last} we may assume that before, between and after the non-degenerate diblocks there are $O(k)$ degenerate diblocks. Thus, the height of
$\hat T$ is $O(k^2)$. Therefore, by Lemma~\ref{thm:boundedd}, the time
complexity for Theorem~\ref{thm:main} is $O^*(2^{O(k^2\log^2 k)})$.\qed

\section{Conclusion}\label{sec:conclusion}
We showed that the \textsc{Rooted $k$-Distinct Branchings} problem is FPT
for general digraphs parameterized by $k$, thereby settling open question
of Bang-Jensen {\em et al.} \cite{bangA76}. The solution
in particular uses a new digraph decomposition, the \emph{rooted cut
  decomposition}, that we believe might be useful for settling other
problems as well.

We did not attempt to optimize the running time of the algorithm of
Theorem~\ref{thm:main}. Perhaps, a more careful handling of degenerate
diblocks may lead to an algorithm of running time $O^*(2^{O(k\log^2 k)})$.
Another question of interest is whether the \textsc{Rooted $k$-Distinct Branchings} 
problem admits a polynomial kernel.

\end{document}